\documentclass[12pt]{article}
\usepackage[a4paper,margin=1in,footskip=0.25in]{geometry}
\usepackage{amssymb,amsmath,amsthm}
\usepackage{bbm} 
\usepackage{enumitem}
\usepackage{caption}
\usepackage{subcaption}	
\usepackage{graphicx}
\usepackage[flushleft]{threeparttable}
\usepackage{siunitx}
\usepackage{booktabs}
\usepackage{xcolor}
\usepackage[colorlinks,citecolor=blue,urlcolor=blue,bookmarks=false,hypertexnames=true]{hyperref}
\usepackage{natbib}
\usepackage[title]{appendix}
\usepackage{xr} 
\newcommand\numberthis{\addtocounter{equation}{1}\tag{\theequation}} 

\newtheorem{assume}{Assumption}

\newtheorem{condition}{Condition}

\newtheorem{prop}{Proposition}
\newtheorem{coro}{Corollary}

\begin{document}
	
	\def\spacingset#1{\renewcommand{\baselinestretch}%
		{#1}\small\normalsize} \spacingset{1.5}
	
	
		\title{\bf Synthetic Controls with Multiple Outcomes\thanks{The authors thank Barton Lee, Stanley Cho, Artem Prokhorov, Tatsushi Oka, Tue Gorgens, Firmin Doko Tchatoka, and seminar participants at UNSW, University of Sydney, the Australia and New Zealand Econometrics Study Group (ANZESG2020) at Monash University, and University of California Irvine, for their constructive suggestions and feedback. Tian acknowledges that this research was supported by the Australian Government Research Training Program (RTP) Scholarship.}}

		\author{Wei Tian\thinspace$^{\textrm{a}}$
			\qquad
			Seojeong Lee\thinspace$^{\textrm{b}}$\thanks{Corresponding author: s.jay.lee@snu.ac.kr}
			\qquad
			Valentyn Panchenko\thinspace$^{\textrm{c}}$}
		\date{\today}
		\maketitle

	\small
	\begin{center}
		{\small $^{\text{a,c}}$~\emph{School of Economics, University of New South Wales, Sydney NSW Australia} \\
			$^{\text{b}}$ \emph{Department of Economics, Seoul National University, Seoul, South Korea}}
	\end{center}
	
	\bigskip

	\begin{abstract}
We generalize the synthetic control (SC) method to a multiple-outcome framework, where the conventional pre-treatment time dimension is supplemented with the extra dimension of related outcomes in computing the SC weights.
This generalization improves the reliability of treatment effect estimation, and can be particularly useful for evaluating the effect of a treatment on multiple outcomes or when only a small number of pre-treatment periods are available.
To illustrate our method, we provide a new perspective on the classic SC application to the 1990 German reunification.
	\end{abstract}
	
	\noindent
	{\it Keywords:}  Synthetic control; Policy evaluation; Causal inference\\
    
	\noindent
	{\it JEL codes:}  C32, C54, I18
	\newpage
	\spacingset{2} 

	\section{Introduction}
	
	The Synthetic Control (SC) method \citep{abadie2003economic,abadie2010synthetic,abadie2015comparative,abadie2021jel} is a popular method for estimating the effect of a policy or intervention on an aggregate unit, such as a country or a city.\footnote{The SC methods were recently used to assess the impacts of the intervention policies during the COVID-19 pandemic \citep{tian2021effects,wang2022measuring}.} The procedure constructs a SC unit as a convex combination of the control units to minimize the difference between the outcome trajectories of the treated unit and the synthesized unit before the treatment, and then estimates the per-period treatment effects using the differences between the observed outcome of the treated unit and the counterfactual SC outcome after the treatment.
	
	Our major contribution is methodological.
	We generalize the conventional single-outcome SC method to a multiple-outcome framework by supplementing the time dimension with the extra dimension of related outcomes. This in effect incorporates additional pre-treatment matching variables and reduces the risk of overfitting.
	We show that the bias of the multiple-outcome SC method is inversely proportional to not only the number of pre-treatment periods as in \cite{abadie2010synthetic}, but also the number of related outcome variables. Hence, the bias is of a smaller order than in the conventional single-outcome SC estimator, which implies that the multiple-outcome SC estimator can be less biased than the conventional SC estimator using the same number of pre-treatment periods.
	
	In addition to the theoretical advantage, the multiple-outcome SC estimator has important practical advantages. Since the SC weights are calculated by matching both the related outcome variables and the pre-treatment periods, it can be applied even when the number of pre-treatment periods is  small.\footnote{Some previous studies have used the conventional SC method by matching on a single outcome with a small number of pre-treatment periods (e.g., \citealp{billmeier2013assessing} and \citealp{cavallo2013catastrophic}), which may produce biased estimates due to overfitting \citep{abadie2015comparative}.} The multiple-outcome SC can be used when we only observe multiple related outcomes in a single pre-treatment period. This is demonstrated in the 1990 German reunification application in Section~\ref{Ch2_sec_illust} and the treatment backdating exercise in the empirical application.
	Furthermore, our method can alleviate the concern on matching over a long period of time, which might be subject to misspecification error due to structural breaks in the relationship between the outcome of interest and the underlying predictors \citep{abadie2021jel}.\footnote{\cite{NBERw31942} raise a related point of emphasizing periods closer to the treatment more than distant periods. Our method addresses this by allowing the SC method to be used with a smaller number of more recent pre-treatment periods.}
	
	Many recent empirical studies analyze multiple outcomes using the SC method, see, e.g.,   \cite{ampofo2022oil,absher2020economic,chu2019joint}. However, they treat each outcome as a separate single-outcome SC application and do not realize the  benefits of matching on multiple outcomes. \cite{cattaneo2021prediction}  mention that researchers may want to match on multiple outcomes, but they do not investigate this further as their focus is on prediction intervals.
	\cite{klossner2018synthesizing} do account for inter-dependencies among multiple outcomes, although they use vector autoregressive models whereas we use factor models. Consequently, their method requires a large number of pre-treatment periods even with multiple outcomes, whereas our method allows a small number of pre-treatment periods in the presence of multiple related outcomes.
	
	A recent work by \cite{sunetal2023} adopts a similar setting to ours and analyzes the bias bounds for SC methods using multiple outcome variables. They compare separate weights for each outcome (the conventional SC), concatenated weights for multiple outcomes (the same as our approach), and average weights for the averaged outcomes. Their novel finding is that the average weights obtained from the averaged standardized outcome variables can reduce bias due to imperfect pre-treatment fit compared to both our approach and the conventional SC. However, as the authors note, equal-weighted averaging may reduce signals present in each of the outcome variables in some cases. Therefore, our approach provides a simple yet valid solution to jointly analyze multiple outcomes and complements the average weights approach of \cite{sunetal2023}.

	The rest of the paper is organized as follows.
	Section~\ref{Ch2_sec_theory} describes the theoretical framework for the multiple-outcome SC method.
	Section~\ref{Ch2_sec_sim} compares the multiple-outcome SC method with the conventional single-outcome SC method using Monte Carlo simulations. Section~\ref{Ch2_sec_illust} compares the methods further by replicating the 1990 German reunification analysis.
	Section \ref{Ch2_sec_con} concludes and provides additional guidance for applied researchers on using our method.
   	The technical details, proofs and extension to the treatment effect on untreated are collected in Appendix~A.
	The practical aspects of estimation and inference together with an illustrative empirical application studying the impacts of non-pharmaceutical interventions during the COVID-19 pandemic are available in the \href{https://drive.google.com/file/d/1LQvfM-_omq-LiCMpkI_iNxkGFUUKKBTT/view?usp=sharing}{Online Appendix}.

	\section{Theoretical Framework}\label{Ch2_sec_theory}
	
	\subsection{Multiple outcomes framework}
	
	Suppose that we observe $K$ outcomes in domain $\mathbb{K}=\{1,2,\dots,K\}$ for $J+1$ units over $T$ time periods, where a domain refers to a collection of related outcomes driven by the same set of observed and unobserved predictors (or factors). For example, the economic domain contains different measures of the economic performance, such as GDP, industrial production, retail sales, and CPI, which can be assumed to depend on the same set of underlying predictors such as infrastructure, technology, natural resources, demographic composition, work ethic, etc.
	
	Without loss of generality, we assume that the first unit ($i=1$) receives the treatment at period $T_0+1\le T$ and remains treated afterwards, while all the other $J$ units ($i=2,\dots,J+1$) are untreated throughout the window of observation.
	Denoting the binary treatment status for unit $i$ at time $t$ as $D_{it}$, we have $D_{it}=1$ for $i=1$ and $t>T_0$, and $D_{it}=0$ otherwise. We consider fixed $J$, large $T$ and $K$ asymptotics.

	We are interested in the effect of the treatment on a single or multiple outcomes in domain $\mathbb{K}$ for the treated unit after the treatment:
	\begin{align}\label{Ch2_eq_ITE}
		\tau_{1t,k}=Y_{1t,k}^1-Y_{1t,k}^0,\enspace t>T_0,\ k\in\mathbb{K},
	\end{align}
	where $Y_{1t,k}^1$ is the potential outcome under the treatment, and $Y_{1t,k}^0$ is the potential outcome without the treatment, so that the observed outcome can be written as $Y_{1t,k}=D_{1t}Y_{1t,k}^1+\left(1-D_{1t}\right)Y_{1t,k}^0$.\footnote{As in \cite{abadie2010synthetic}, we treat $\tau_{1t,k}$ as given once the sample is drawn.}
	Since we only observe the treated potential outcome but not the untreated potential outcome for unit 1 at $t>T_0$, we need to predict the counterfactual outcome $Y_{1t,k}^0$. 
	
	Suppose that the untreated potential outcome $k\in\mathbb{K}$ for unit $i$ at time $t$ is given by an interactive fixed effects model
	\begin{align}\label{Ch2_eq_IFE0}
		Y_{it,k}^0=\delta_{t,k}+\boldsymbol{Z}_i'\boldsymbol{\theta}_{t,k}+\boldsymbol{\mu}_i'\boldsymbol{\lambda}_{t,k}+\varepsilon_{it,k},
	\end{align}
	where $\delta_{t,k}$ is the time trend in outcome $k$, $\boldsymbol{Z}_i$ and $\boldsymbol{\mu}_i$ are the $r\times 1$ and $f\times 1$ vectors of observed and unobserved predictors of $Y_{it,k}^0$ with outcome-specific coefficients $\boldsymbol{\theta}_{t,k}$ and $\boldsymbol{\lambda}_{t,k}$, respectively, and $\varepsilon_{it,k}$ is the idiosyncratic transitory shock. The model in \eqref{Ch2_eq_IFE0} is similar to that of \cite{abadie2010synthetic} and \cite{abadie2021jel} but is allowed to have the outcome-specific constant and coefficients of observable and unobservable predictors. 
	
	Assuming that the individual-specific unobserved predictor $\boldsymbol{\mu}_{i}$ is common to the outcomes in the same domain is the key assumption in our model. If the outcomes depend on different sets of unobserved predictors (thus $\boldsymbol{\mu}_{i,k}$), then we lose the benefit of matching on multiple related outcomes in terms of having the same order of bias with the conventional SC method. Nevertheless, we argue that our assumption is no different from the standard factor analysis where a low number of common factors underlying related variables is assumed. Moreover, our model accommodates the case when the unobserved predictors are separable: $\boldsymbol{\mu}_{i,k}=\boldsymbol{\mu}_{i} + \boldsymbol{u}_{k}$ where $\boldsymbol{u}_{k}$ is outcome-specific predictors. 
	
	A few more remarks on the model. The model does not exclude the possibility that some outcomes in the domain depend on other predictors that are serially uncorrelated and independent from the included predictors and the treatment status, which can thus be treated as part of the transitory shocks. In addition, the coefficients may contain zero so that the corresponding predictors may affect some outcomes in some periods, but not all outcomes in all periods, as long as there is enough variation in the coefficients across different pre-treatment periods or outcomes, as specified in Condition \ref{Ch2_assume_rank}.
	
	There is strong suggestive evidence for our model. We show in Section~\ref{Ch2_sec_illust}, that the SC for West Germany constructed by matching on multiple economic variables just in the single year of 1989, can track the trajectory of West Germany's GDP closely for 30 years, and the trajectory for the multiple-outcome SC after the treatment is also very similar to that of the SC constructed by matching on 30 years of pre-treatment GDP.

We impose technical conditions similar to those in \cite{abadie2010synthetic} and \cite{botosaru2019role}. 

			\begin{condition}[]\label{Ch2_assume_error} Transitory shocks  
			\leavevmode
			\begin{enumerate}[label=\arabic*)]
				\item $\varepsilon_{it,k}$ are independent across $i$, $t$, $k$;\footnote{Note that despite the i.i.d. condition on $\varepsilon_{it,k}$, the unobserved interactive fixed effects may account for the correlations along the corresponding dimensions. In practice, the outcomes may have different scales or volatilities so that the transitory shocks may be clustered at the outcome level. This complexity can be dealt with by standardizing each outcome in each period before matching.}
				\item $\mathbbm{E}(\varepsilon_{it,k}\mid \boldsymbol{Z}_{j},\boldsymbol{\mu}_j)=0$ for all $i$, $j$, $t$ and $k$;
				\item $\mathbb{E}|\varepsilon_{jt,k}|^p<\infty$ for all $j=2,\dots,J+1$, $t\le T_0$, $k\in\mathbb{K}$ and some even integer $p\ge2$.
			\end{enumerate}
		\end{condition}			
							

	\begin{condition}[]\label{Ch2_assume_rank}
	The smallest eigenvalue of $\frac{1}{KT_0}\sum_{k=1}^{K}\sum_{t=1}^{T_0}\boldsymbol{\lambda}_{t,k}\boldsymbol{\lambda}_{t,k}'$ is bounded from below by some positive number $\underline{\xi}$.
\end{condition}

	A SC is constructed using a convex combination of the control units such that the SC matches the treated unit in terms of the observed predictors and the pre-treatment values of the $K$ related outcomes. This can be achieved if the matching variables of the treated unit is in the convex hull of those of the control units.
	
	\begin{condition}[]\label{Ch2_assume_perfect_fit}
		There exists a set of weights $\left(\hat{w}_2,\dots,\hat{w}_{J+1}\right)$ such that $\hat{w}_j\ge 0$ for $j=2,\dots,J+1$, $\sum_{j=2}^{J+1}\hat{w}_j=1$, and for all $t\le T_0$ and $k\in\mathbb{K}$     
    \begin{equation}
  \sum_{j=2}^{J+1} \hat{w}_j\boldsymbol{Z}_j = \boldsymbol{Z}_1, \qquad
  \sum_{j=2}^{J+1} \hat{w}_jY_{jt,k}  =  Y_{1t,k}.
    \label{eq:pretreat_fit}
    \end{equation}
	\end{condition}

  In practice, this condition may hold only approximately and the SC weights are obtained by minimizing a weighted sum of the squared distance between the left-hand and the right-hand sides of \eqref{eq:pretreat_fit}.

The multiple-outcome SC estimator for $\tau_{1t,k}$ is then constructed as
	\begin{align}\label{Ch2_eq_SCM}
		\widehat{\tau}_{1t,k}=Y_{1t,k}-\sum_{j=2}^{J+1} \hat{w}_jY_{jt,k}.
	\end{align}

 	Condition~3  is analogous to the perfect fit condition of \cite{abadie2010synthetic} (their Eq.~2), and is used to show the asymptotic unbiasedness of the SC estimator. This condition is unarguably strong, but some recent papers provide justification for the SC method when the condition holds only approximately.  \cite{botosaru2019role} relax the perfect fit condition for the observed predictors. \cite{ferman2021synthetic} show that under typical large $T$ asymptotics, the bias of the conventional SC estimator does not disappear as the number of pre-treatment periods increases if the perfect fit condition does not hold. Nonetheless, they show that a demeaned version of the SC estimator can still perform better than the difference-in-difference estimator. \cite{ferman2021properties} relaxes the condition by using large $T$, large $J$ asymptotics. 
	More recently, \cite{zwz2023} show the asymptotic optimality of the SC estimator under imperfect pre-treatment fit, in the sense that it achieves the lowest possible squared prediction error among all possible treatment effect estimators that are based on an average of control units, and this asymptotic optimality continues to hold without assuming a linear factor model. This finding provides confidence that the SC method is still the best choice among all estimators of a similar construction under imperfect pre-treatment fit.
	Although we do not attempt an extension of our multiple-outcome SC to a more general setting in this paper, we believe that a similar optimality can be established under suitable conditions.

	To facilitate a direct comparison between the multiple-outcome SC estimator and the conventional single-outcome SC estimator based only on the $k$th outcome, let $\{\tilde{w}_j^{(k)}\}_{j=2}^{J+1}$ be the single-outcome SC weights such that $\tilde{w}_j^{(k)}\ge 0$ for $j=2,\dots,J+1$, $\sum_{j=2}^{J+1}\tilde{w}_j^{(k)}=1$, $\sum_{j=2}^{J+1} \tilde{w}_j^{(k)}\boldsymbol{Z}_j=\boldsymbol{Z}_1$ and $\sum_{j=2}^{J+1} \tilde{w}_j^{(k)}Y_{jt,k}=Y_{1t,k}$ for all $t\le T_0$.
	The conventional SC estimator for $\tau_{1t,k}$ is 
	\begin{align}\label{Ch2_eq_SCM_single}
		\widetilde{\tau}_{1t,k}=Y_{1t,k}-\sum_{j=2}^{J+1} \tilde{w}_j^{(k)}Y_{jt,k}.
	\end{align}
	
	The following proposition shows that the bias of the multiple-outcome SC method is reducing faster than that of the conventional single-outcome SC method. 
	
	\begin{prop}\label{Ch2_prop_1}
		Suppose that $J$ is fixed, whereas $T_0$ and $K$ are increasing, and the technical conditions are satisfied, then, for any $t>T_0$,
		\begin{align*}
			\left\vert \mathbb{E}(\widetilde{\tau}_{1t,k})-\tau_{1t,k}\right\vert & =O\left(\frac{1}{\sqrt{T_{0}}}\right),  \\
			\left\vert \mathbb{E}(\widehat{\tau}_{1t,k})-\tau_{1t,k}\right\vert   & =O\left(\frac{1}{\sqrt{KT_{0}}}\right).
		\end{align*}
	\end{prop}

	Since the bias of the multiple-outcome SC estimator reduces as the number of pre-treatment periods and the number of related outcomes increases, in practice we can use the SC method even when the number of pre-treatment periods is  small, if multiple related outcomes are available and the treated unit can be closely approximated by the SC for these outcomes in the pre-treatment periods. In Appendix~A, we compare the  multiple-outcome SC weights with the single-outcome SC weights and find that they are generically different.  

	The estimation and inference procedures are similar to those of the conventional SC method \citep*{abadie2021jel}, by including the multiple related outcomes in the pre-treatment periods as matching variables (see their practical implementation for the empirical application in the \href{https://drive.google.com/file/d/1LQvfM-_omq-LiCMpkI_iNxkGFUUKKBTT/view?usp=sharing}{Online Appendix}).

	\subsection{Adjusting for differences in levels}\label{Ch2_sec_levels}

The conventional SC method requires the treated unit to be in the convex hull of the control units in terms of the pre-treatment matching variables. However, there are cases where the values of the matching variables are extreme for the treated unit, such that no convex combination of the control units can closely approximate the outcome of the treated unit in the pre-treatment periods. In particular, it is often the case in practice that there are relatively stable differences in the level of the outcomes across units before the treatment.
In such cases, \cite{ferman2021synthetic} and \cite{abadie2021jel} suggest constructing the SC using demeaned outcomes, i.e., outcomes measured in differences with respect to their pre-treatment means. This enables the SC to track the dynamics in the outcome of the treated unit over time, while allowing the levels to differ by a constant amount, which is similar to the ``parallel trends'' assumption in the difference-in-differences method.\footnote{Using demeaned outcomes is also similar to a proposal in \cite{doudchenko2017}, which includes an intercept when minimizing the difference between the SC and the treated unit in the matching variables.} Apart from allowing a better pre-treatment fit for the treated unit, using demeaned outcomes has the additional merit that it helps correct the size distortion of the permutation test based on the post-to-pre-treatment RMSPE (Root Mean Squared Prediction Error) ratios obtained from permuting the treatment status among all units.


In the multiple-outcome framework, the relative position of the units can vary across different outcomes, making it difficult to match on multiple outcomes simultaneously. The demeaning process would be helpful in these circumstances by improving the pre-treatment fits.

Provided that $T_{0}\geq2$, replacing the outcomes, $Y_{is,k}$ in \eqref{eq:pretreat_fit} with demeaned outcomes $\dot{Y}_{it,k}=Y_{it,k}-\frac{1}{T_0}\sum_{s=1}^{T_0} Y_{is,k}$, we obtain a new set of weights, $\left(\hat{w}_2,\dots,\hat{w}_{J+1}\right)$. We can then construct the multiple-outcome SC estimator for $\tau_{1t,k}$ using the demeaned outcomes as
\begin{align}\label{Ch2_eq_SCM1}
	\widehat{\tau}_{1t,k}^{\text{DM}}=\dot{Y}_{1t,k}-\sum_{j=2}^{J+1} \hat{w}_j\dot{Y}_{jt,k}
\end{align}
to account for the differences in the level of the outcomes. Similarly with Proposition \ref{Ch2_prop_1}, the bias of this estimator can be shown to shrink as the number of related outcomes and pre-treatment periods goes to infinity.

\begin{coro}\label{Ch2_coro_1}
	Suppose that $J$ is fixed, whereas $T_0$ and $K$ are increasing, the demeaned outcomes are used, and the technical conditions are satisfied, then for any $t>T_0$,
	\begin{align*}
		\left\vert \mathbb{E}(\widehat{\tau}_{1t,k}^{\text{DM}})-\tau_{1t,k}\right\vert & =O\left(\frac{1}{\sqrt{KT_{0}}}\right).
	\end{align*}
\end{coro}


	\section{Monte Carlo Simulations}\label{Ch2_sec_sim}
	
We conduct Monte Carlo simulations to compare the multiple-outcome SC method and the conventional single-outcome SC method. 	
	The number of post-treatment periods is fixed at 1, and the number of units at 30, with a single treated unit and 29 control units.
	The treatment effect is set to 0, so that the treated potential outcomes are the same with the untreated potential outcomes.\footnote{The zero treatment effect is set to investigate the size of the test. It does not affect the bias and the standard deviation of the SC estimator.}
	The data generating process (DGP) is as follows:
	\begin{align}\label{Ch2_eq_IFE_sim}
		Y_{it,k}=\delta_{t,k}+\boldsymbol{Z}_i'\boldsymbol{\theta}_{t,k}+\boldsymbol{\mu}_i'\boldsymbol{\lambda}_{t,k}+\varepsilon_{it,k},
	\end{align}
	where $\boldsymbol{Z}_i$ is the vector of 2 observed predictors and $\boldsymbol{\mu}_i$ is the vector of 4 unobserved predictors. 
	The observed and unobserved predictors, $\boldsymbol{Z}_i$ and $\boldsymbol{\mu}_i$, are drawn independently from the uniform distribution $U[-1,1]$ for the control units, and $U[-d,d]$ with $d\in[0,1]$ for the treated unit.
	When $d=1$, the treated unit is equally likely to obtain extreme values in the outcomes as the control units. When $d$ is smaller, the treated unit is more likely to be in the convex hull of the control units, and thus have better pre-treatment fits.
	
	To generate outcomes that are closer to real data, where there are often clear differences in the level of the outcomes across units and that the levels are relatively stable over time, we set the variance of the mean of the coefficients to be large relative to the variance of the coefficients and the transitory shocks. As such, the time trend ($\delta_{t,k}$) and the coefficients ($\boldsymbol{\theta}_{t,k}$ and $\boldsymbol{\lambda}_{t,k}$) are drawn independently from the normal distribution $N(\omega_k,1)$ with $\omega_k\sim~N(0,10^2)$, and the transitory shocks, $\varepsilon_{it,k}$, are drawn independently from the standard normal distribution.\footnote{When $\omega_k$ is a constant, there would be no distinguishable levels in the outcomes for different units, and demeaning the outcomes would not be useful in this case.}
	
	Figure \ref{Ch2_fig_sim} displays the trajectories for two of the related outcomes from a typical simulated sample with $d=1$ and $T_0=10$. The trajectories for the treated unit are in black, and the control units in gray. As intended, there are visible and stable differences in the level of the outcomes, and the levels for the treated unit are different across different outcomes. Although the two outcomes share the same underlying predictors, they may appear different due to differences in the time trends, the coefficients on the predictors and the transitory shocks. Overall, the trajectories of the outcomes generated in our simulation look very similar to those in real data, e.g., as seen in the descriptive graphs for the outcomes in the empirical application (Figure B.2 in the \href{https://drive.google.com/file/d/1LQvfM-_omq-LiCMpkI_iNxkGFUUKKBTT/view?usp=sharing}{Online Appendix}).
	\begin{figure}[!h]
		\centering
		\includegraphics[width=.7\linewidth]{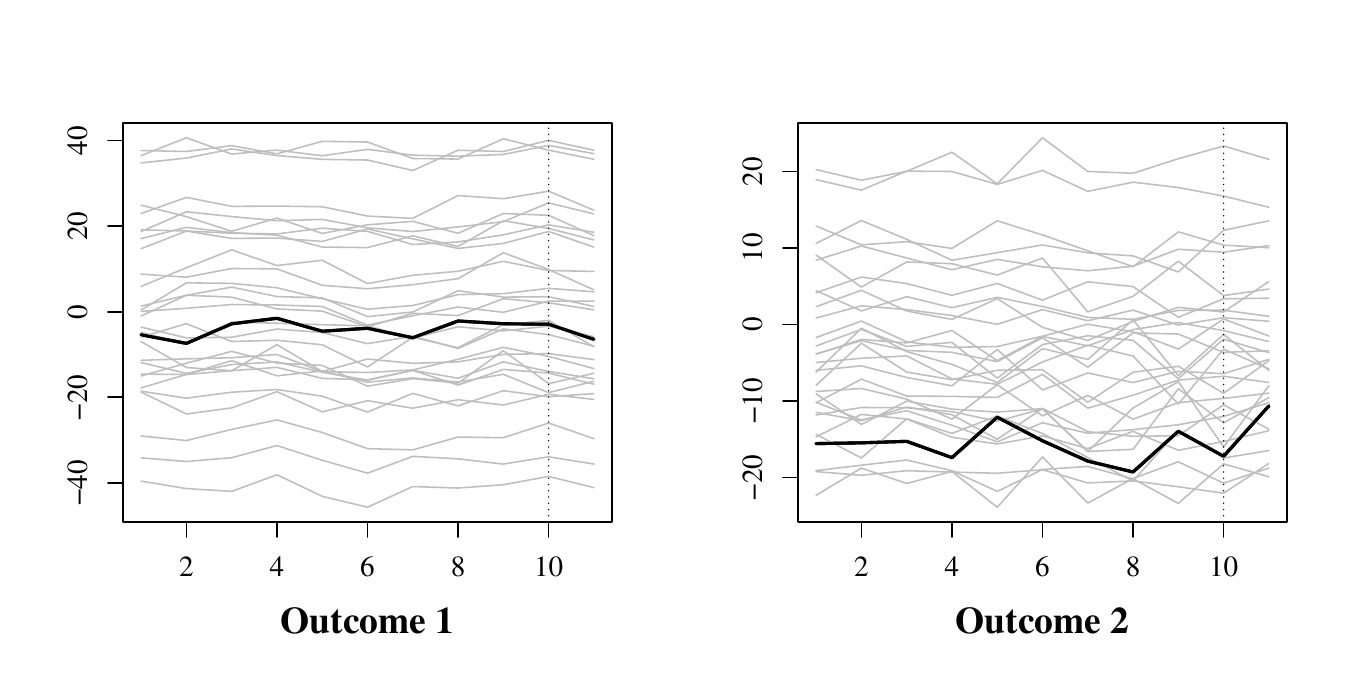}
		\caption{A Simulated Example}
		\label{Ch2_fig_sim}
	\end{figure}

	There are 5000 simulations for each pair of $d$ and $T_0$. In each simulation, the observed and unobserved predictors are drawn only once so that the outcomes are related by sharing the same underlying predictors, while the time trend, the coefficients and the transitory shocks are drawn independently across different outcomes.
	We compare the conventional single-outcome SC estimator and the multiple-outcome SC estimators constructed using $K=1$, 3, and 10 demeaned outcomes, respectively. The only difference between the single-outcome SC method and the multiple-outcome SC method when $K=1$ is the use of the demeaned outcomes. To measure their performances, we estimate the treatment effect on outcome 1 for the treated unit at $t=T_0+1$, and compute the average absolute bias and standard deviation of the estimators as well as the average rejection rate of the 10\% test in 5000 simulations. The null hypothesis of zero treatment effect is rejected in each simulation, if the RMSPE ratio for the treated unit is ranked among the largest 10\%, i.e., top 3 in our sample, in the permutation test.
	
	When the estimation improves with better pre-treatment fits and larger numbers of pre-treatment periods and related outcomes, we expect the average absolute bias and standard deviation of the estimators to be closer to $\sqrt{\frac{2}{\pi}}\approx0.8$ and 1, which are the mean and standard deviation of the standard normal distribution folded at the mean (half-normal distribution).\footnote{Note that the terms from the bias decomposition in the proof of Proposition \ref{Ch2_prop_1} are all close to 0, except the post-treatment transitory shock, which follows a standard normal distribution in our simulation.} When the distributions of the RMSPE ratio for the treated unit and the control units are close, so that there is little size distortion in the permutation test, we expect the average rejection rate of the 10\% test to be close to the nominal rejection rate at 10\%. The pre-treatment RMSPE is also reported as a measure of pre-treatment fit.

Several findings emerge from the results of the simulations, which are reported in Table~\ref{Ch2_tab_sim}.
First, when the number of pre-treatment periods or related outcomes increases, the bias decreases and becomes closer to the expected value for all estimators. Similar patterns are observed for the standard deviation of the estimators. Furthermore, the results clearly show that matching on more pre-treatment variables alleviates overfitting and improves out-of-sample prediction. Note that overfitting is a finite-sample problem due to erroneously including transitory shocks in the matching and should be distinguished from condition \eqref{eq:pretreat_fit} which is assumed to hold for a large $T_{0}$ and $K$.

Second, when the support of the predictors for the treated unit is the same with that for the control units ($d=1$), demeaning substantially improves estimation in terms of both bias and standard deviation, as the conventional single-outcome SC method is likely to perform poorly when the treated unit is far from the convex hull of the control units, whereas demeaning adjusts for the differences in the levels and improves the pre-treatment fit. Meanwhile, the rejection rate of the 10\% test is close to the nominal size, with or without demeaning in this case, since the RMSPE ratio for the treated unit is not conditional on a good pre-treatment fit.
When $d$ is smaller, the probability of obtaining a good pre-treatment fit increases for the treated unit while staying unchanged for the other units. As a result, the bias and standard deviation for both the conventional SC estimator and the demeaned SC estimator are smaller, and the improvement in estimation by demeaning the outcomes becomes less pronounced. In contrast, the distortion in the size of the test increases drastically, and demeaning alleviates the size distortion by improving the pre-treatment fits for all units.

Third, a larger number of pre-treatment periods or related outcomes also reduces the size distortion, since the pre-treatment RMSPE for the treated unit is less likely to be very close to 0.
Overall, the results show that the multiple-outcome SC method outperforms the conventional single-outcome method, when there are multiple related outcomes with stable differences in the level of the outcomes.\\

			\begin{threeparttable}[t]
				\centering
				\setlength{\tabcolsep}{3pt}  
				\resizebox{16.5cm}{!}{
				\begin{tabular}{cc@{\hskip 0.5ex}ccccc @{\hskip 0.5ex}ccccc @{\hskip 0.5ex}ccccc @{\hskip 0.5ex}ccccc}
					\toprule
					\multicolumn{3}{c}{} & \multicolumn{4}{c}{Conventional SC} & \multicolumn{1}{c}{} & \multicolumn{4}{c}{Multi-Outcome SC} & \multicolumn{1}{c}{} & \multicolumn{4}{c}{Multi-Outcome SC} & \multicolumn{1}{c}{} & \multicolumn{4}{c}{Multi-Outcome SC}                                                                  \\
					\multicolumn{3}{c}{} & \multicolumn{4}{c}{}                & \multicolumn{1}{c}{} & \multicolumn{4}{c}{($K=1$)}          & \multicolumn{1}{c}{} & \multicolumn{4}{c}{($K=3$)}          & \multicolumn{1}{c}{} & \multicolumn{4}{c}{($K=10$)}                                                                           \\
					\cmidrule(lr){4-7} \cmidrule(lr){9-12} \cmidrule(lr){14-17} \cmidrule(lr){19-22}
					$d$                  & $T_0$                               &                      & Pre-fit & Bias                                 & SD                   & Rej.                                 &                     & Pre-fit & Bias                                 & SD   & Rej.  &  & Pre-fit & Bias & SD   & Rej.  & & Pre-fit & Bias & SD   & Rej.  \\
					\midrule
					\addlinespace[2ex]
					1 & 5 &  & 1.65 & 1.94 & 2.91 & 0.10 &  & 0.51 & 1.43 & 1.81 & 0.10 &  & 0.82 & 1.32 & 1.67 & 0.10 &  & 0.99 & 1.22 & 1.54 & 0.10 \\ 
  1 & 10 &  & 1.63 & 1.64 & 2.47 & 0.10 &  & 0.83 & 1.27 & 1.61 & 0.10 &  & 1.04 & 1.19 & 1.50 & 0.10 &  & 1.14 & 1.12 & 1.40 & 0.10 \\ 
  1 & 20 &  & 1.62 & 1.52 & 2.36 & 0.10 &  & 1.03 & 1.18 & 1.49 & 0.10 &  & 1.15 & 1.11 & 1.41 & 0.10 &  & 1.20 & 1.08 & 1.36 & 0.10 \\
					\addlinespace[2ex]
					0.5 & 5 &  & 0.44 & 1.10 & 1.40 & 0.36 &  & 0.23 & 1.16 & 1.47 & 0.32 &  & 0.56 & 1.08 & 1.36 & 0.15 &  & 0.77 & 1.01 & 1.26 & 0.12 \\ 
					0.5 & 10 &  & 0.71 & 1.03 & 1.29 & 0.24 &  & 0.54 & 1.08 & 1.35 & 0.19 &  & 0.80 & 1.01 & 1.26 & 0.14 &  & 0.91 & 0.95 & 1.18 & 0.12 \\ 
					0.5 & 20 &  & 0.86 & 0.95 & 1.20 & 0.17 &  & 0.77 & 0.99 & 1.25 & 0.15 &  & 0.92 & 0.92 & 1.16 & 0.12 &  & 0.99 & 0.89 & 1.11 & 0.10 \\
					\addlinespace[2ex]
					0 & 5 &  & 0.24 & 1.05 & 1.32 & 0.57 &  & 0.15 & 1.09 & 1.37 & 0.48 &  & 0.48 & 1.04 & 1.31 & 0.19 &  & 0.71 & 0.99 & 1.23 & 0.13 \\ 
  0 & 10 &  & 0.54 & 0.98 & 1.23 & 0.34 &  & 0.45 & 1.03 & 1.29 & 0.25 &  & 0.72 & 0.96 & 1.20 & 0.15 &  & 0.86 & 0.90 & 1.13 & 0.13 \\ 
  0 & 20 &  & 0.73 & 0.92 & 1.16 & 0.23 &  & 0.68 & 0.96 & 1.21 & 0.18 &  & 0.86 & 0.90 & 1.13 & 0.14 &  & 0.93 & 0.87 & 1.09 & 0.12 \\ 
					\bottomrule
				\end{tabular}
				}
				\caption{Simulation comparison of the pre-treatment fit, average absolute bias, standard deviation, and rejection rate of the 10\% test for the single-outcome SC estimator, and the multiple-outcome SC estimators constructed using 1, 3 and 10 demeaned outcomes respectively, with varying $d$ and $T_0$, based on 5000 simulations for each setting.}\label{Ch2_tab_sim}
			\end{threeparttable}

\section{Empirical illustration: the 1990 German reunification}\label{Ch2_sec_illust}

The multiple-outcome SC method extends the applicability of the popular SC method to cases where only a small number of pre-treatment periods are available. Here we provide an empirical illustration that the results produced by matching on multiple related outcomes in a few pre-treatment periods are similar to those by matching on a single outcome in many pre-treatment periods.

We re-analyze the effect of the 1990 German reunification on West Germany's GDP per capita \citep{abadie2015comparative}. This example is ideal, because not only is the outcome of interest observed in many pre-treatment periods, but also numerous outcomes in the economic domain are available from the OECD statistics.
For comparison, we construct two SCs for West Germany, one by matching on the annual GDP per capita from 1960 to 1990, the other by matching on multiple related outcomes in 1989 only.\footnote{The list of outcomes in 1989 includes private social expenditure, energy supply per GDP, electricity generation, triadic patent families, real GDP growth, CPI, trade openness, total tax revenue, and GDP per capita (see Table~B.1 in the \href{https://drive.google.com/file/d/1LQvfM-_omq-LiCMpkI_iNxkGFUUKKBTT/view?usp=sharing}{Online Appendix}). The results produced by matching on multiple related outcomes in any single year from 1985-1989 are very similar.}

\begin{figure}[!htbp]
	\centering
	\includegraphics[width=.7\linewidth]{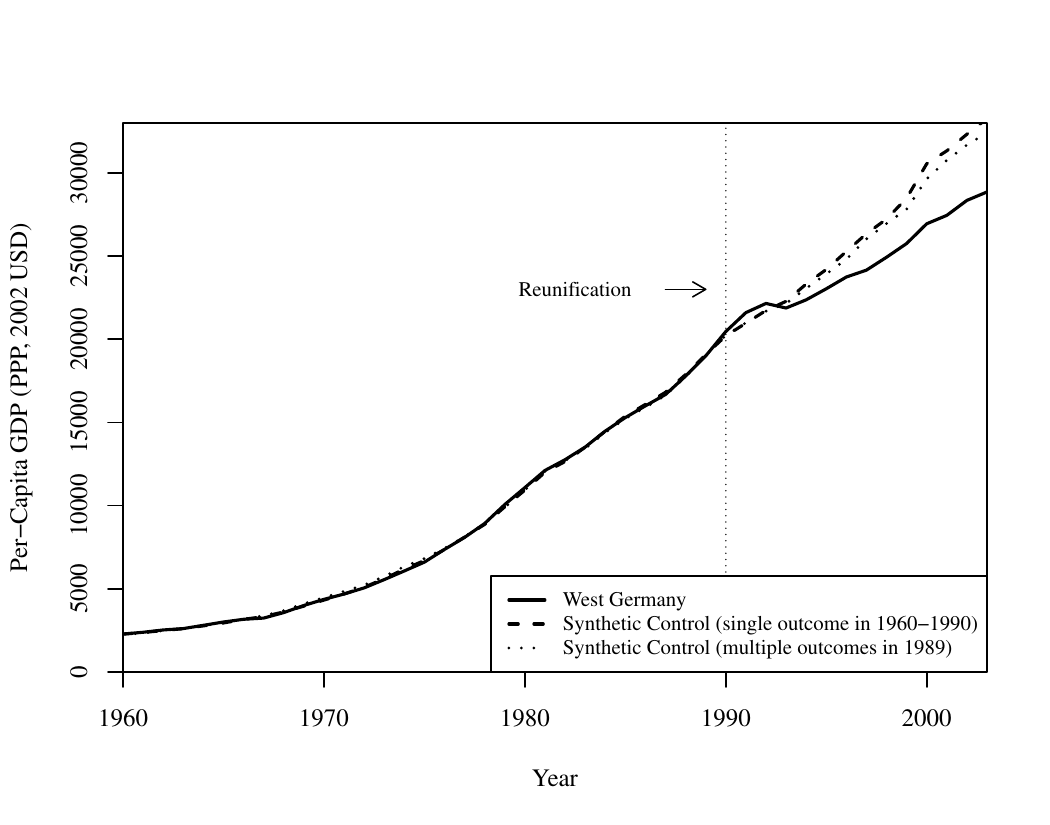}
	\caption{Re-analysis of the Economic Cost of the 1990 German Reunification}
	\label{Ch2_fig_Germany}
\end{figure}

Figure \ref{Ch2_fig_Germany} compares the trajectories of GDP per capita for West Germany and the two SCs. As we would hope for, both SCs are able to track West Germany's trajectory closely over a span of 30 years prior to the treatment.
This is not surprising for the single-outcome SC, which is constructed with the aim to match the values of the outcome observed over the 30 years as closely as possible.
The fact that the synthetic West Germany constructed using multiple outcomes only in the year of 1989 can track the trajectory of GDP per capita in West Germany so closely for so long demonstrates the ability of the multiple-outcome SC method to produce a SC that closely approximates the treated unit in terms of the underlying predictors, even when the number of pre-treatment periods is as small as one.
The estimated treatment effects, represented by the gaps between the trajectories of the realized outcome and the counterfactual outcomes, are also similar for the two SC estimators after 1990.

	\section{Conclusion}\label{Ch2_sec_con}
	
	This paper generalizes the conventional single-outcome SC method to a multiple-outcome framework, where the number of pre-treatment periods is supplemented with the number of related outcomes, improving the SC methods in  the situations when the number of available pre-treatment periods is small or when potential structural breaks impose limitations on the usable number of pre-treatment periods. We show that the bias of the multiple-outcome SC estimator diminishes with the number of pre-treatment periods and the number of related outcome variables. Our simulations show that the multiple-outcome SC estimator improves upon the conventional single-outcome SC estimator in terms of bias, standard deviation, and the size.

	Finally, let us offer additional guidance for applied researchers who might benefit from our method. 
	Just as the original SC method of \cite{abadie2015comparative} is not recommended ``when the pre-treatment fit is poor or the number of pre-treatment periods is small'', our method is best suited for applications where
	1) there are enough pre-treatment variables to match on (i.e., there is either a large number of related outcomes or  pre-treatment periods, or both), so overfitting is unlikely,\footnote{The number of pre-treatment variables required to avoid overfitting depends on the size of the donor pool. In the usual cases where the number of control units (donors) is in dozens, preferably at least 10 pre-treatment variables are needed.} and
	2) the SC constructed using a few control units can closely approximate the pre-treatment variables of the treated unit, which is unlikely to happen with a large number of outcome variables if they do not share the common linear structure.\footnote{There has not been an accepted rule of thumb to determine whether the pre-treatment fit is good other than relying on visual inspections. \cite{gardeazabal2017empirical} consider matches that have a pre-treatment mean absolute error smaller than 20\% of the mean values in the outcome variable as good matches, but this may not apply to cases where the mean of the outcome is close to 0.}

	\begin{appendices}\label{print_appendix}
		
		
		\section{Technical Derivations and Proofs}\label{Ch2_sec_proof}

		\setcounter{equation}{0}
		\renewcommand\theequation{\thesection.\arabic{equation}}		
		
		\subsection*{Comparison of multi-outcome and single-outcome SC weights}
		
 For better understanding of the multiple-outcome SC, we compare the  multiple-outcome SC weights with the single-outcome SC weights. Let $\boldsymbol{w}$ be the $J\times1$ vector and $\boldsymbol{Y}_{t,k}^{c}$ be the $J\times1$ vector of the outcome variable of the control units. Ignoring the facts that the elements of $\boldsymbol{w}$ sum to one and nonnegative, the conventional single-outcome SC weights for the $k$th outcome are obtained by minimizing
		\begin{equation}
		\sum_{t=1}^{T_{0}}(Y_{1,t,k}-\boldsymbol{w}'\boldsymbol{Y}_{t,k}^{c})^2.
			\label{crit1}
		\end{equation}
		The solution $\boldsymbol{\widetilde{w}}_{k}$ is simply the least squares estimator:
		\begin{equation}
			\boldsymbol{\widetilde{w}}_{k} = \left(\sum_{t=1}^{T_{0}}\boldsymbol{Y}_{t,k}^{c}\boldsymbol{Y}_{t,k}^{c'}\right)^{-1}\sum_{t=1}^{T_{0}}\boldsymbol{Y}_{t,k}^{c}Y_{1,t,k},
			\label{sol1}
		\end{equation}
		provided that $\sum_{t=1}^{T_{0}}\boldsymbol{Y}_{t,k}^{c}\boldsymbol{Y}_{t,k}^{c'}$ is invertible. For the multiple-outcome SC weights, we further stack $\boldsymbol{Y}_{t,k}^{c}$ over the $K$ outcomes. Then the criterion function is
		\begin{equation}
			\sum_{k=1}^{K}\sum_{t=1}^{T_{0}}(Y_{1,t,k}-\boldsymbol{w}'\boldsymbol{Y}_{t,k}^{c})^2
			\label{crit2}
		\end{equation}
		and the solution is
		\begin{equation}
			\boldsymbol{\widehat{w}} = \left(\sum_{k=1}^{K}\sum_{t=1}^{T_{0}}\boldsymbol{Y}_{t,k}^{c}\boldsymbol{Y}_{t,k}^{c'}\right)^{-1}\sum_{k=1}^{K}\sum_{t=1}^{T_{0}}\boldsymbol{Y}_{t,k}^{c}Y_{1,t,k},
			\label{sol2}
		\end{equation}
		provided that $\sum_{k=1}^{K}\sum_{t=1}^{T_{0}}\boldsymbol{Y}_{t,k}^{c}\boldsymbol{Y}_{t,k}^{c'}$ is invertible. The actual implementation of the SC method restricts the weights to sum one and to be non-negative, so the solutions will be different from \eqref{sol1} and \eqref{sol2}. Nonetheless, the least-squares estimators \eqref{sol1} and \eqref{sol2} show that $\boldsymbol{\widehat{w}}$ is not a simple average of $\boldsymbol{\widetilde{w}}_{k}$ across the $K$ outcomes. In addition, the minimized criterion functions are exactly zero, so both $\boldsymbol{\widetilde{w}}$ and $\boldsymbol{\widehat{w}}$ give the optimal fit for each outcome variable.

		\subsection*{Proofs}

		\begin{proof}[Proof of Proposition \ref{Ch2_prop_1}]


			To provide a unified framework for the biases of the multiple-outcome SC estimator and the single-outcome SC estimator, without loss of generality, we can write
			\begin{align}\label{Ch2_eq_weight}
				\tilde{w}_j^{(k)}= \hat{w}_j +\tilde{v}_j^{(k)},\ j=2,\dots,J+1,
			\end{align}
			where $-1\le \tilde{v}_j^{(k)}\le 1$.
			We emphasize that both sets of weights satisfy \eqref{eq:pretreat_fit} for both the observed predictors and the pre-treatment values of the $k$th outcome.
			
			Under the restrictions $\sum_{j=2}^{J+1} \hat{w}_j\boldsymbol{Z}_j=\boldsymbol{Z}_1$, the bias of the multiple-outcome SC estimator for outcome $k$ and $t>T_{0}$ is the expectation of
			\begin{align*}
				Y_{1t,k}^0-\sum_{j=2}^{J+1} \hat{w}_jY_{jt,k} & =\left(\boldsymbol{\mu}_1-\sum_{j=2}^{J+1} \hat{w}_j\boldsymbol{\mu}_j\right)'\boldsymbol{\lambda}_{t,k}+\varepsilon_{1t,k}-\sum_{j=2}^{J+1} \hat{w}_j\varepsilon_{jt,k}, \numberthis\label{Ch2_eq_e}
			\end{align*}
			and the bias of the single-outcome SC estimator is the expectation of
			\begin{align*}
				Y_{1t,k}^0 - \sum_{j=2}^{J+1}\tilde{w}_j^{(k)}Y_{jt,k} = & \left(\boldsymbol{\mu}_{1}-\sum_{j=2}^{J+1}\tilde{w}_j^{(k)}\boldsymbol{\mu}_{j}\right)'\boldsymbol{\lambda}_{t,k} + \varepsilon_{1t,k}-\sum_{j=2}^{J+1}\tilde{w}_j^{(k)}\varepsilon_{jt,k}                                                                                                           \\
				=                                                        & \left(\boldsymbol{\mu}_{1}-\sum_{j=2}^{J+1}\hat{w}_j\boldsymbol{\mu}_{j}\right)'\boldsymbol{\lambda}_{t,k} -\sum_{j=2}^{J+1}\tilde{v}_j^{(k)}\boldsymbol{\mu}_{j}'\boldsymbol{\lambda}_{t,k}+ \varepsilon_{1t,k}-\sum_{j=2}^{J+1}\tilde{w}_j^{(k)}\varepsilon_{jt,k}. \numberthis \label{Ch2_eq_e_SO}
			\end{align*}
			
			We can stack observations over the pre-treatment periods to replace the first term on the RHS of \eqref{Ch2_eq_e} and the first two terms on the RHS of \eqref{Ch2_eq_e_SO} with functions of the coefficients and the individual transitory shocks from the stacked expressions.
			Stacking the pre-treatment outcomes ${Y}_{it,k}$ over the $T_{0}$ pre-treatment periods, we have
			\begin{equation}
				\boldsymbol{Y}_{i,k} = \boldsymbol{\delta}_{k} + \boldsymbol{\theta}_{k}\boldsymbol{Z}_{i} + \boldsymbol{\lambda}_{k}\boldsymbol{\mu}_{i} + \boldsymbol{\varepsilon}_{i,k},
				\label{Ch2_eq_stack}
			\end{equation}
			where $\boldsymbol{Y}_{i,k}$, $\boldsymbol{\delta}_{k}$, and $\boldsymbol{\varepsilon}_{i,k}$ are $T_{0}\times 1$, and $\boldsymbol{\theta}_{k}$ and $\boldsymbol{\lambda}_{k}$ are $T_{0}\times r$ and $T_{0}\times F$, respectively.
			
			Since $\sum_{j=2}^{J+1} \tilde{w}_j^{(k)}\boldsymbol{Z}_j=\boldsymbol{Z}_1$, the restrictions $\sum_{j=2}^{J+1}\tilde{w}_j^{(k)}\boldsymbol{Y}_{j,k}=\boldsymbol{Y}_{1,k}$ can be simplified to
			\begin{align}
				\boldsymbol{\lambda}_{k}\left(\boldsymbol{\mu}_{1}-\sum_{j=2}^{J+1}\tilde{w}_j^{(k)}\boldsymbol{\mu}_{j}\right)=\sum_{j=2}^{J+1}\tilde{w}_j^{(k)}\boldsymbol{\varepsilon}_{j,k}-\boldsymbol{\varepsilon}_{1,k}.\label{Ch2_eq_simp_SO}
			\end{align}
			Similarly, using the multiple-outcome SC weights, we have
			\begin{align}
				\boldsymbol{\lambda}_{k}\left(\boldsymbol{\mu}_{1}-\sum_{j=2}^{J+1}\hat{w}_j\boldsymbol{\mu}_{j}\right)=\sum_{j=2}^{J+1}\hat{w}_j\boldsymbol{\varepsilon}_{j,k}-\boldsymbol{\varepsilon}_{1,k},\label{Ch2_eq_simp_MO}
			\end{align}
			which further simplifies \eqref{Ch2_eq_simp_SO} to
			\begin{align}
				-\boldsymbol{\lambda}_{k}\sum_{j=2}^{J+1}\tilde{v}_j^{(k)}\boldsymbol{\mu}_{j}=\sum_{j=2}^{J+1}\tilde{v}_j^{(k)}\boldsymbol{\varepsilon}_{j,k}.\label{Ch2_eq_simp1_SO}
			\end{align}
			
			Since the $K$ outcomes are determined by the same set of predictors in our multiple-outcome framework, we can further stack \eqref{Ch2_eq_stack} over the $K$ outcomes to get
			\begin{align}
				\boldsymbol{Y}_{i}=\boldsymbol{\delta}+\boldsymbol{\theta} \boldsymbol{Z}_i+\boldsymbol{\lambda}\boldsymbol{\mu}_i+\boldsymbol{\varepsilon}_{i},
			\end{align}
			where $\boldsymbol{Y}_{i}$, $\boldsymbol{\delta}$ and $\boldsymbol{\varepsilon}_{i}$ are $KT_0\times 1$, and $\boldsymbol{\theta}$ and $\boldsymbol{\lambda}$ are $KT_0\times r$ and $KT_0\times f$, respectively.
			And the restrictions $\sum_{j=2}^{J+1} \hat{w}_j\boldsymbol{Y}_j=\boldsymbol{Y}_1$ can be simplified to
			\begin{align}\label{Ch2_eq_simp}
				\boldsymbol{\lambda}\left(\boldsymbol{\mu}_1-\sum_{j=2}^{J+1} \hat{w}_j\boldsymbol{\mu}_j\right)=\sum_{j=2}^{J+1} \hat{w}_j\boldsymbol{\varepsilon}_{j}-\boldsymbol{\varepsilon}_{1}.
			\end{align}
			
			Condition \ref{Ch2_assume_rank} states that the $f\times f$ matrix $\boldsymbol{\lambda}'\boldsymbol{\lambda}$ has full rank, thus pre-multiplying $\left(\boldsymbol{\lambda}'\boldsymbol{\lambda}\right)^{-1}\boldsymbol{\lambda}'$ on both sides of \eqref{Ch2_eq_simp}, we have
			\begin{align}\label{Ch2_eq_simp1}
				\left(\boldsymbol{\mu}_1-\sum_{j=2}^{J+1} \hat{w}_j\boldsymbol{\mu}_j\right)=\left(\boldsymbol{\lambda}'\boldsymbol{\lambda}\right)^{-1}\boldsymbol{\lambda}'\left(\sum_{j=2}^{J+1} \hat{w}_j\boldsymbol{\varepsilon}_{j}-\boldsymbol{\varepsilon}_{1}\right),
			\end{align}
			so that \eqref{Ch2_eq_e} can be written as
			\begin{align*}
				Y_{1t,k}^0-\sum_{j=2}^{J+1} \hat{w}_jY_{jt,k}= & \boldsymbol{\lambda}_{t,k}'\left({\boldsymbol{\lambda}}'\boldsymbol{\lambda}\right)^{-1}{\boldsymbol{\lambda}}'\sum_{j=2}^{J+1} \hat{w}_j\boldsymbol{\varepsilon}_{j} \tag{$B_{1t,k}$}                                         \\
				& -\boldsymbol{\lambda}_{t,k}'\left({\boldsymbol{\lambda}}'\boldsymbol{\lambda}\right)^{-1}{\boldsymbol{\lambda}}'\boldsymbol{\varepsilon}_{1}+\varepsilon_{1t,k}-\sum_{j=2}^{J+1} \hat{w}_j\varepsilon_{jt,k}. \tag{$B_{2t,k}$}
			\end{align*}
			
			Similarly, pre-multiplying $\boldsymbol{\lambda}_{t,k}'(\boldsymbol{\lambda}_{k}'\boldsymbol{\lambda}_{k})^{-1}\boldsymbol{\lambda}_{k}'$ to both sides of \eqref{Ch2_eq_simp1_SO} gives
			\begin{equation}
				-\boldsymbol{\lambda}_{t,k}'\sum_{j=2}^{J+1}\tilde{v}_j^{(k)}\boldsymbol{\mu}_{j} = \boldsymbol{\lambda}_{t,k}'(\boldsymbol{\lambda}_{k}'\boldsymbol{\lambda}_{k})^{-1}\boldsymbol{\lambda}_{k}'\sum_{j=2}^{J+1}\tilde{v}_j^{(k)}\boldsymbol{\varepsilon}_{j,k}.
				\label{Ch2_stack4}
			\end{equation}
			
			Replacing \eqref{Ch2_eq_simp1} and \eqref{Ch2_stack4} into \eqref{Ch2_eq_e_SO}, we have
			\begin{align*}
				Y_{1t,k}^0-\sum_{j=2}^{J+1}\tilde{w}_j^{(k)}Y_{jt,k}
				= & \boldsymbol{\lambda}_{t,k}'(\boldsymbol{\lambda}'\boldsymbol{\lambda})^{-1}\boldsymbol{\lambda}'\sum_{j=2}^{J+1}\hat{w}_j\boldsymbol{\varepsilon}_{j} \tag{$B_{1t,k}$} \label{B1}                                       \\
				& -\boldsymbol{\lambda}_{t,k}'(\boldsymbol{\lambda}'\boldsymbol{\lambda})^{-1}\boldsymbol{\lambda}'\boldsymbol{\varepsilon}_{1}+\varepsilon_{1t,k}-\sum_{j=2}^{J+1}\hat{w}_j\varepsilon_{jt,k} \tag{$B_{2t,k}$} \label{B2} \\
				& -\sum_{j=2}^{J+1}\tilde{v}_j^{(k)}\varepsilon_{jt,k} \tag{$B_{3t,k}$}   \label{B3}                                                                                                                                      \\
				& +\boldsymbol{\lambda}_{t,k}'(\boldsymbol{\lambda}_{k}'\boldsymbol{\lambda}_{k})^{-1}\boldsymbol{\lambda}_{k}'\sum_{j=2}^{J+1}\tilde{v}_j^{(k)}\boldsymbol{\varepsilon}_{j,k}. \tag{$B_{4t,k}$}\label{B4}
			\end{align*}
			
			Given Condition~\ref{Ch2_assume_error} and since the SC weights are independent of the observations for $t>T_0$, $B_{2t,k}$ and $B_{3t,k}$ have zero mean, whereas $B_{1t,k}$ and $B_{4t,k}$ do not because $\hat{w}_j$ and $\tilde{w}_j^{(k)}$ are functions of $\boldsymbol{\varepsilon}_{j}$ \citep{botosaru2019role}.
			
			It is shown that $\mathbb{E}\left\vert B_{4t,k}\right\vert=O\left(\frac{1}{\sqrt{T_{0}}}\right)$ in \cite{abadie2010synthetic}.
			Following closely the proof in Appendix B of \cite{abadie2010synthetic}, we show that $\mathbb{E}\left\vert B_{1t,k}\right\vert=O\left(\frac{1}{\sqrt{KT_{0}}}\right)$.
			We can rewrite $B_{1t,k}$ as
			\begin{align}
				B_{1t,k}=\sum_{j=2}^{J+1} \hat{w}_j\sum_{q=1}^{K}\sum_{s=1}^{T_0}\boldsymbol{\lambda}_{t,k}'\left(\sum_{l=1}^{K}\sum_{n=1}^{T_0}\boldsymbol{\lambda}_{n,l}\boldsymbol{\lambda}_{n,l}'\right)^{-1}\boldsymbol{\lambda}_{s,q}\varepsilon_{js,q}.
			\end{align}
			
			Let the largest element of $|\boldsymbol{\lambda}_{t,k}|$ for $t=1,\dots,T$ and $k=1,\dots,K$ be bounded from above by $\bar{\lambda}$.
			Under Condition \ref{Ch2_assume_rank} and using the Cauchy–Schwarz Inequality, we have
			\begin{align*}
				& \left(\boldsymbol{\lambda}_{t,k}'\left(\sum_{l=1}^{K}\sum_{n=1}^{T_0}\boldsymbol{\lambda}_{n,l}\boldsymbol{\lambda}_{n,l}'\right)^{-1}\boldsymbol{\lambda}_{s,q}\right)                                                                                                                                                                                                    \\
				\le & \left(\boldsymbol{\lambda}_{t,k}'\left(\sum_{l=1}^{K}\sum_{n=1}^{T_0}\boldsymbol{\lambda}_{n,l}\boldsymbol{\lambda}_{n,l}'\right)^{-1}\boldsymbol{\lambda}_{t,k}\right)^{\frac{1}{2}}\left(\boldsymbol{\lambda}_{s,q}'\left(\sum_{l=1}^{K}\sum_{n=1}^{T_0}\boldsymbol{\lambda}_{n,l}\boldsymbol{\lambda}_{n,l}'\right)^{-1}\boldsymbol{\lambda}_{s,q}\right)^{\frac{1}{2}} \\
				\le & \left(\frac{\bar{\lambda}^2f}{KT_0\underline{\xi}}\right).
			\end{align*}
			
			Let $\bar{\varepsilon}_j=\sum_{q=1}^{K}\sum_{s=1}^{T_0}\boldsymbol{\lambda}_{t,k}'\left(\sum_{l=1}^{K}\sum_{n=1}^{T_0}\boldsymbol{\lambda}_{n,l}\boldsymbol{\lambda}_{n,l}'\right)^{-1}\boldsymbol{\lambda}_{s,q}\varepsilon_{js,q}$.
			Then by H\"{o}lder's Inequality and the norm monotonicity, we have
			$$|B_{1t,k}| \le \sum_{j=2}^{J+1} \hat{w}_j|\bar{\varepsilon}_j| \le \left(\sum_{j=2}^{J+1} |\hat{w}_j|^q\right)^{1/q}\left(\sum_{j=2}^{J+1} |\bar{\varepsilon}_j|^p\right)^{1/p}\le \left(\sum_{j=2}^{J+1} |\bar{\varepsilon}_j|^p\right)^{1/p},$$
			with $p,q>1$ and $\frac{1}{p}+\frac{1}{q}=1$.
			
			Using H\"{o}lder's Inequality again, we have
			$$\mathbb{E}\left[\sum_{j=2}^{J+1} |\bar{\varepsilon}_j|\right] \le \mathbb{E}\left[\left(\sum_{j=2}^{J+1} |\bar{\varepsilon}_j|^p\right)^{1/p}\right] \le \left(\mathbb{E}\left[\sum_{j=2}^{J+1} |\bar{\varepsilon}_j|^p\right]\right)^{1/p}=\left(\sum_{j=2}^{J+1} \mathbb{E}|\bar{\varepsilon}_j|^p\right)^{1/p}.$$
			
			Then using Rosenthal’s Inequality, we have
			$$\mathbb{E}|\bar{\varepsilon}_j|^p \le C\left(p\right)\left(\frac{\bar{\lambda}^2f}{KT_0\underline{\xi}}\right)^p\text{max}\left\{\sum_{q=1}^{K}\sum_{s=1}^{T_0}\mathbb{E}|\varepsilon_{js,q}|^p,\left(\sum_{q=1}^{K}\sum_{s=1}^{T_0}\mathbb{E}|\varepsilon_{js,q}|^2\right)^{p/2}\right\},$$
			where the constant $C\left(p\right)=\mathbb{E}\left(\phi-1\right)^p$ with $\phi$ being a Poisson random variable with parameter 1.
			
			Let $\bar{m}_p=\text{max}_j \frac{1}{KT_0}\sum_{q=1}^{K}\sum_{s=1}^{T_0}\mathbb{E}|\varepsilon_{js,q}|^p$, then we have
			\begin{align}
				\mathbb{E}|B_{1t,k}|\le C\left(p\right)^{1/p}\left(\frac{\bar{\lambda}^2f}{\underline{\xi}}\right)J^{1/p}\text{max}\left\{\frac{\bar{m}_p^{1/p}}{\left(KT_0\right)^{1-1/p}},\frac{\bar{m}_2^{1/2}}{{\left(KT_0\right)}^{1/2}}\right\}.\label{Ch2_eq_bound}
			\end{align}
			
			Therefore, $\mathbb{E}\left\vert B_{1t,k}\right\vert=O\left(\frac{1}{\sqrt{KT_{0}}}\right)$, and $\mathbb{E}\left(\widehat{\tau}_{1t,k}-\tau_{1t,k}\right)\rightarrow 0 \enspace \text{as} \enspace KT_0\rightarrow \infty,$
			i.e., the bias of the multiple-outcome SC estimator is bounded by a function that goes to zero when the number of outcomes in the domain or the pre-treatment periods goes to infinity.
			
			We have shown that the bias of the conventional single-outcome SC method is usually $O\left(\frac{1}{\sqrt{T_{0}}}\right)$, but if the single-outcome SC weights coincide with the multiple-outcome SC weights, in which case $\tilde{v}_j^{(k)}=0$ and $\tilde{w}_j^{(k)}=\hat{w}_j$, then the order becomes $O\left(\frac{1}{\sqrt{KT_{0}}}\right)$.
			
		\end{proof}

		\begin{proof}[Proof of Corollary \ref{Ch2_coro_1}]
			
			The bias for the demeaned SC estimator is
			\begin{align*}
				\widehat{\tau}_{1t,k}^{\text{DM}}-\tau_{1t,k} & =\dot{Y}_{1t,k}-\sum_{j=2}^{J+1} \hat{w}_j\dot{Y}_{jt,k}-Y_{1t,k}^1+Y_{1t,k}^0                                       \\
				& =Y_{1t,k}^1-\frac{1}{T_0}\sum_{s=1}^{T_0}{Y}_{1s,k}-\sum_{j=2}^{J+1} \hat{w}_j\dot{Y}_{jt,k}^0-Y_{1t,k}^1+Y_{1t,k}^0 \\
				& =\dot{Y}_{1t,k}^0-\sum_{j=2}^{J+1} \hat{w}_j\dot{Y}_{jt,k}^0.
			\end{align*}
			
			Notice that the demeaned equation for $Y_{it,k}^0$ retains the interactive fixed effects structure:
			\begin{align*}
				\dot{Y}_{it,k}^0 & =Y_{it,k}^0-\frac{1}{T_0}\sum_{s=1}^{T_0} Y_{is,k}^0                                                                                               \\
				& =\dot{\delta}_{t,k}+\boldsymbol{Z}_i'\dot{\boldsymbol{\theta}}_{t,k}+\boldsymbol{\mu}_i'\dot{\boldsymbol{\lambda}}_{t,k}+\dot{\varepsilon}_{it,k}.
			\end{align*}
			
			We can thus follow similar steps to show that $\mathbb{E}\left(\widehat{\tau}_{1t,k}^{\text{DM}}-\tau_{1t,k}\right)\rightarrow 0 \enspace \text{as} \enspace KT_0\rightarrow \infty$ under the demeaned outcomes, and conditions~\ref{Ch2_assume_error} and \ref{Ch2_assume_rank}.
			
		\end{proof}

		
		\subsection*{Treatment effect on the untreated}\label{Ch2_sec_untreated}
		
		In the discussion of the theoretical framework, we have been focusing on the setting with a single treated unit and many control units, where we estimate the treatment effects on the treated. However, there are cases with many treated units, and we may wish to estimate the treatment effects on the untreated by constructing a synthetic unit for the untreated unit using the treated units.
		
		Without loss of generality, suppose that unit 1 remains untreated within the window of observation, while all the other units are treated from $t=T_0+1$ onwards.
		Recall that the treated potential outcome is $Y_{it,k}^1=Y_{it,k}^0+\tau_{it,k}$. Since we have not imposed any assumption on the treatment effects except treating them as fixed given the sample, the treated potential outcomes may not have the interactive fixed effects functional forms or depend on the same predictors as the untreated potential outcomes.
		As a consequence, a synthetic unit that matches the untreated unit in the pre-treatment matching variables may not credibly reproduce its counterfactual outcomes after the treatment, even if it is similar to the untreated unit in the underlying predictors of the untreated potential outcomes.
		Therefore, in order to estimate the treatment effects on the untreated, we assume that the treatment effects are determined by the same predictors of the untreated potential outcomes, that is\footnote{Note, this assumption is more general than assuming that the treatment effects are constant. A similar assumption is discussed in \cite{athey2021matrix}, where the treatment effect is assumed to have a low-rank pattern.}	
		\begin{align}
			\tau_{it,k}=\alpha_{t,k}+\boldsymbol{Z}_i'\boldsymbol{\beta}_{t,k}+\boldsymbol{\mu}_i'\boldsymbol{\gamma}_{t,k}, \enspace \forall i,t,k,
			\label{eq:assume_ITE}
		\end{align}
		where $\alpha_{t,k}$ is the time trend in $\tau_{it,k}$, and $\boldsymbol{\beta}_{t,k}$ and $\boldsymbol{\gamma}_{t,k}$ are the outcome-specific coefficients of the observed and unobserved predictors respectively.

		The multiple-outcome SC estimator for $\tau_{1t,k}$, the treatment effect on the untreated unit, can then be constructed as
		\begin{align}
			\check{\tau}_{1t,k}=\sum_{j=2}^{J+1} \hat{w}_jY_{jt,k}-Y_{1t,k}.
		\end{align}
		
		Note that Eq.~\ref{eq:assume_ITE}, together with Eq.~1 and Eq.~2 in the main text, implies that the treated potential outcome $k$ for unit $i$ at time $t$ is given by
		\begin{align}\label{Ch2_eq_IFE1}
			Y_{it,k}^1=\left(\delta_{t,k}+\alpha_{t,k}\right)+\boldsymbol{Z}_i'\left(\boldsymbol{\theta}_{t,k}+\boldsymbol{\beta}_{t,k}\right)+\boldsymbol{\mu}_i'\left(\boldsymbol{\lambda}_{t,k}+\boldsymbol{\gamma}_{t,k}\right)+\varepsilon_{it,k}.
		\end{align}
		Since the treated potential outcome $Y_{it,k}^1$ has an interactive fixed effects structure, we can similarly show that the bias of $\check{\tau}_{1t,k}$ vanishes as the number of pre-treatment periods or related outcomes grows.
		
		\begin{coro}\label{Ch2_coro_2}
			Suppose that $J$ is fixed, whereas $T_0$ and $K$ are increasing, the demeaned outcomes are used, \eqref{eq:assume_ITE} holds, and the technical conditions are satisfied, then for any $t>T_0$,
			\begin{align*}
				\left\vert \mathbb{E}(\check{\tau}_{1t})-\tau_{1t,k}\right\vert & =O\left(\frac{1}{\sqrt{KT_{0}}}\right).
			\end{align*}
		\end{coro}


		\begin{proof}[Proof of Corollary \ref{Ch2_coro_2}]
			The proof is similar to that of Proposition 1 and thus omitted.
		\end{proof}

	\end{appendices}

	\bigskip
	\begin{center}
		{\large\bf SUPPLEMENTARY MATERIALS}
	\end{center}
	
	\begin{description}
		
		\item[Online appendix:] The online appendix contains the data source and full analysis of the empirical application, details of the estimation and inference procedures, and additional results. (\href{https://drive.google.com/file/d/1LQvfM-_omq-LiCMpkI_iNxkGFUUKKBTT/view?usp=sharing}{.pdf file})
		
		\item[Datasets and R codes:] Datasets and R codes used for the simulations and the empirical application, as well as a readme file that describes the contents. (\href{https://drive.google.com/file/d/1GkwlrGoP5R0f4I3i3bifSv83r5_Ji8Jg/view?usp=sharing}{.zip file})
		
	\end{description}

	\bibliography{References2}
	\bibliographystyle{apalike}
	
\end{document}